\documentclass[11pt]{article}
\usepackage[]{amsmath,amssymb,amsfonts,latexsym,amsthm,enumerate,fullpage}
\newcommand{\remove}[1]{}

\newcommand{\trace}{\mathrm{Trace}}
\newcommand{\ra}{\rightarrow}

\def\F{\mathbb{F}}
\def\N{\mathbb{N}}

\newtheorem{definition}{Definition}

\newtheorem{thm}{Theorem}
\newtheorem{cor}{Corollary}

\newtheorem{claim}[thm]{Claim}
\newtheorem{lemma}[thm]{Lemma}
\newtheorem{proposition}[thm]{Proposition}
\newcommand{\eqdef}{{\stackrel{\rm def}{=}}}
\begin{document}
\title{Edge Transitive Ramanujan Graphs and \\Highly Symmetric LDPC Good Codes}
\author{Tali Kaufman \thanks {Bar-Ilan University, ISRAEL. Email: \texttt{kaufmant@mit.edu}.
Research supported in part by the Alon Fellowship.}
 \and Alexander Lubotzky \thanks {Hebrew University, ISRAEL. Email: \texttt{alexlub@math.huji.ac.il}.
Research supported in part by the ERC and by the Israel Science Foundation.}}
\maketitle

\begin{abstract}
We present a symmetric LDPC code with constant rate and constant distance (i.e. good LDPC code) that its constraint space is generated
by the orbit of one constant weight constraint under a group action. Our construction provides the first symmetric LDPC good codes.
This solves the main open problem raised by Kaufman and Wigderson in~\cite{KW}.
\end{abstract}

\section{Introduction}
An $(n,k,d)$-code is a subspace $C$ of $\F_2^X$, where $|X|=n$, of dimension $k$ such that the (Hamming) weight of every vector
$0 \neq v \in C$ is at least $d$. A code (or rather a family of codes, when $n \rightarrow \infty$) is called {\em good} if there exists
an $\epsilon > 0$ such that $r(C)=\frac{k}{n}$ (the rate) and $\delta(C)=\frac{d}{n}$ (normalized distance) are both at least $\epsilon$.
For a code $C \subseteq \F_2^X$ we denote by $C^{\bot}$ its dual (i.e. all vectors "orthogonal" to $C$) and we think of its vectors
as the constraints defining the code (i.e. they define linear functionals on $\F_2^X$ whose common set of solutions is $C$). The code $C$ is called {\em LDPC} if there exists a set of defining constraints of bounded weight. This bound is called the density of the code.

The code $C$ is said to be {\em symmetric} (w.r.t. $H$) if there exists a group $H$ acting {\em transitively} on $X$ such that the induced action on $\F_2^X$ preserves $C$. The code $C$ is called {\em single-orbit symmetric} if in addition there exists $v \in C^{\bot}$ such that $C^{\bot}$ is spanned by the orbit $H \cdot v$ (i.e. $C$ is defined by the equation $x \cdot v = 0$ and its translations by $H$). We say that $C$ is {\em highly symmetric} if furthermore the vector $v$ can be chosen to have a bounded weight (which, in particular, implies that $C$ is an LDPC).
Many of the codes studied in coding theory are symmetric and even single-orbit symmetric (though not necessarily highly-symmetric), e.g., all the cyclic codes (see section~\ref{section:cyclic} below). But, unfortunately cyclic codes do not tend to have the other desired good properties. For example, it is a long standing conjecture that there are no good cyclic codes. An old result of Berman from $1967$~\cite{BERMAN} proved it for infinitely many code lengths. In~\cite{BSS05} Babai, Shpilka and Stefankovic show that cyclic good LDPC codes do not exist.


In~\cite{KW} Kaufman and Wigderson initiated the study of highly symmetric LDPC codes. The reader is referred to their paper for motivation.
The main question presented there is:
{\bf "To what extent can symmetric LDPC codes attain (or even come close to) the coding theory gold standards of
linear distance and constant rate?"}

The authors consider the tradeoff between $1/$rate and the density in symmetric codes. In the codes known prior to their work, if one was constant then the other was worst possible. They constructed a symmetric code with better tradeoff between $1/$rate and the density, namely a symmetric code with constant rate, nearly constant distance, whose density is poly logarithmic in the code length. Thus, their code satisfy some of the gold standards but not all. Moreover, they show that if the group $H$ is abelian or solvable (of bounded derived length) there are no codes satisfying all the desired properties and expressed some skepticism if such codes exist at all.

Our main result is on the optimistic side and give:

\begin{thm}[main]\label{theorem:main}
There exist explicit highly symmetric LDPC good codes.
\end{thm}

So, our codes meet all the "gold standards" of coding theory. Our code $C$ have:
\begin{itemize}
\item Constant rate.
\item Constant relative distance.
\item Symmetric under a group action $H$.
\item The dual code $C^{\bot}$ is generated by a single orbit $H \cdot v$.
\item The above $v$ can be chosen to have bounded weight. In particular, $C$ is also LDPC.
\end{itemize}

Our constructions are of Cayley codes, in the framework of~\cite{KW}, but with different groups and different generators.
Cayley codes are defined in~\cite{KW} as follows.
Let $G$ be a group of order $m$ and $S$ a symmetric (i.e. $S=S^{-1}$) set of generators of order $t$. Let $Cay(G,S)$ be the (right) Cayley graph of $G$ w.r.t. $S$ and $E$ its set of edges. So $|E| = mt/2$. Assume $B \subseteq \F_2^S$ is a linear code. Let $C(G,S,B)$ be the linear subspace of $\F_2^E$ containing all the functions $f:E \rightarrow \F_2$ such that for every $g \in G$ the "local view" of $f$ at the star of $g$ is in $B$, i.e. the function $f_g:S \rightarrow \F_2$ given by $f_g(s) := f((g,gs))$ is in $B$.

In general, Cayley codes are {\em not} symmetric. The group $G$ acts transitively on the vertices of $Cay(G,S)$ but not on its edges. The Cayley code is symmetric in the following situation. Let $G$ be a group generated by a symmetric set $S$ (i.e. $S=S^{-1}$) and $T$ a group acting on $G$ (i.e., there exists a homomorphism $\varphi:T \rightarrow Aut(G)$). Assume that $S$ is an orbit of the action, namely, there exists $\gamma \in G$ such that
$$ S = \{ \varphi(\alpha)(\gamma) | \alpha \in T \}.$$
In this case one can show, see Section~\ref{sec:cayley}, that the semi-direct product group $H=G \rtimes T$ acts on $Cay(G,S)$ and this action is transitive on the edges.

Now, the group $T$ acts on $S$ and hence on $\F_2^S$. If the "small code" $B$ is preserved by $T$ (i.e. $B$ is symmetric w.r.t $T$), then $H = G \rtimes T$ preserves $C(G,S,B)$. Moreover, if $B$ is single-orbit symmetric (w.r.t. $T$), e.g., if $T$ is a cyclic group and $B$ is a cyclic code, then $C(G,S,B)$ is also single-orbit symmetric, and in this case it is also automatically highly symmetric when $|S|$ is bounded. (See Section~\ref{sec:cayley} below).

The following theorem (proved in~\cite{KW} inspired by~\cite{Tan81} and~\cite{SS96}), estimates the parameters of $C=C(G,S,B)$ in terms of those of $B$ and the eigenvalues of the graph $Cay(G,S)$.

\begin{thm}\label{thm:Cayley-codes}
Let $G,S,T,B$ as above, with $B$ a code in $\F_2^S$ with rate $r(B)$ and normalized distance $\delta(B)$. Then $C=C(G,S,B)$ is a code with $r(C) \geq 2 r(B) -1$ and $\delta(C) \geq [(\delta(B) - \lambda)/(1 - \lambda)]^2$ where $\lambda = \lambda (Cay(G,S))$ is the second largest normalized eigenvalue of the Cayley graph $Cay(G,S)$.
\end{thm}

\begin{cor}\label{cor:Cayley-codes} In the notations above; if $r(B) > \frac{1}{2}$ and $\delta(B) > \lambda (Cay(G,S))$ then $C(G,S,B)$ is a good code. If in addition $B$ is single-orbit symmetric, then so is $C=C(G,S,B)$. Hence, if $|G| \rightarrow \infty$ and $|S|$ is bounded $C$ is highly symmetric.
\end{cor}

The last corollary gives the framework to prove Theorem~\ref{theorem:main}. We will present first edge transitive Cayley graphs $Cay(G,S)$ (where $S$ is the orbit of a cyclic group $T$ of a fixed size $q+1$ acting on $G$, when $|G| \rightarrow \infty$) with $\lambda = \lambda (Cay(G,S))$ sufficiently small. Secondly, we will find a cyclic code $B \subseteq \F_2^S$ with $r(B) > \frac{1}{2}$ and $\delta(B) > \lambda$. The resulting $C(G,S,B)$ will be highly symmetric and good by Corollary~\ref{cor:Cayley-codes}. In particular it is also LDPC as it is defined by local equations (in fact the orbit of one equation under the group $H=G \rtimes T$) which touches at most $q+1$ variables. So, Theorem~\ref{theorem:main} will be proven once the two goals will be achieved.

For the first mission let us recall that a finite $q+1$-regular connected graph is called Ramanujan if for every normalized eigenvalue $\lambda$ either $|\lambda|=1$ or $|\lambda| \leq  \frac{2\sqrt{q}}{q+1}$. Such graphs where constructed in~\cite{LPS} for every prime $p$ and for every prime power $q=p^{\ell}$ in~\cite{Mor}. But we will make use of a more recent explicit construction of Ramanujan graphs
by Lubotzky, Samuels and Vishne~\cite{LSV2} which have some extra symmetry, and in particular are edge transitive.

\begin{thm}[Edge Transitive Ramanujan graphs Theorem]~\label{theorem:ramanujan}
For a prime power $q$ and for $\alpha \in \N$ such that $q^{\alpha} >17$, let $G=PSL_2(q^{\alpha})$ or $G=PGL_2(q^{\alpha})$ and let $T$ be the non split tori of order $q+1$ in $PGL_2(q)$.
There exists $\gamma \in G$ such that $Cay(G,S)$ is a $q+1$-Ramanujan graph with $S=\{t \gamma t^{-1} | t \in T\}$, $|S|=q+1$, in particular, the graph $Cay(G,S)$ is edge-transitive. The element $\gamma$ will be explicitly defined in section~\ref{section:ramanujan}.
In the case that $G=PGL_2(q^{\alpha})$ the graph $Cay(G,S)$ is bi-partite.
\end{thm}

The above theorem will give us the desired Cayley graphs.
For $B$ we will make a very special choice. Assume now $m \geq 10$, $q = 2^{m+1} - 3$ and assume that $q$ is a prime power.
For example, one can take $m=11$, $q=4093$ (so $q$ is prime in this case).

\begin{thm}[Good $B$ Theorem]\label{theorem:cyclic-code}
Let $q \in \N$ be a prime power such that $q+1 =2^{m+1} - 2 = 2(2^m-1)$ for $m\geq 10$.
Then there exists an explicit linear binary cyclic code $B \subseteq \F_2^{q+1}$ with $r(B)  > \frac{1}{2}$ and $\delta(B) > \frac{2 \sqrt{q}}{q+1}$.
\end{thm}

The proof of Theorem \ref{theorem:cyclic-code} will use standard methods of coding theory. The cyclic codes which are natural to be chosen to have the required rate and distance are BCH-codes. However, a little obstacle is caused by the fact that over $\F_2$, BCH codes are always of odd length, while for our construction we need them to be of even length. This is overcome using a trick from~\cite{VL} (see Section~\ref{section:cyclic} for details).

An interesting number theoretic problem is whether one can find infinitely many $q$'s suitable for us. I.e. are there infinitely many $r$'s for which
$2^{r+1}-3$ is a prime (power)? This is a question of the the same style of the famous {\em Mersanne Primes Problem}; Are there infinitely many primes of the form $2^m -1$? It seems that with the current knowledge we also do not know the answer even if we replace "primes" with "prime powers".
Luckily, we need only one such $q$ and $2^{12}-3 = 4093$ does the job for us!



\section{Notations and Conventions}\label{sec:defs}
We start with some basic definitions that are being used throughout
this work.

\subsection{Group theory definitions}

\begin{definition}[Action of a group on a set, transitivity]~\label{group-set-action}
An {\em action} of a group $T$ on a set $X$ is a group
homomorphism $\phi: T \rightarrow Sym(X)$ that sends each element $t$ to a permutation of the elements of $X$.
Let $x^t$ denotes the action of $t \in T$ on $x \in X$.
That is, an action should satisfy for every $t,t' \in T, x \in X$
$$ x^{t t'} =(x^{t'})^{t}$$

The {\em orbit} of an element $x \in X$ is
$$x^{T} = \{x^{t} | t \in T\}.$$

The action is called {\em transitive} if for some (and hence every) $x \in X$, $x^{T}=X$.
\end{definition}



\begin{definition} [Action of a group on a group]
An {\em action} of a group $T$ on a group $G$ is a group
homomorphism $\phi: T \rightarrow Aut(G)$.
Let $g^t=\phi_{t}(g)$ denotes the action of $t \in T$ on $g \in G$.
\end{definition}

\begin{definition}[Semi-direct product group]
Suppose a group $T$ acts on a group $G$. The {\em semidirect
product} $G \rtimes T$ is a group whose elements are pairs $(g,t)$
where $g \in G$ and $t \in T$, and the product is given by:
$$ (g_1,t_1) \cdot (g_2,t_2) = (g_1 \cdot g_2^{t_1}, t_1 \cdot
t_2).$$
\end{definition}

%

Note that with the identification of $G$ as a subgroup $\{(g,1)\}$ of $G \rtimes T$, we have $g^t = tgt^{-1}$ which respects the equality $(g^{t'})^{t} = g^{t t'}$.

%
%

%
%
%
%
%
%
%
%
%
%

\subsection{Graph definitions}

\begin{definition}[Cayley graph]
Given a group $G$ and a set of generators $S \subset G$ ($S=S^{-1}$), the {\em
Cayley graph} $Cay(G,S)$ is a graph, whose vertices are labeled by
elements of $G$. The edges adjacent to $g \in G$ are $(g,s)$, $s \in S$ and $(g,s) = (gs,s^{-1})$.
%
\end{definition}


\begin{definition}[Edge-transitive graph]
A graph $Y=(V,E)$ is {\em edge-transitive} if $Aut(Y)$ acts transitively on the undirected edges.
\end{definition}

Here is a situation in which the Cayley graph is edge transitive.

\begin{definition}[Action of the semi-direct product group]\label{Def:SDP-action}
Let $G$ be a group and $S \subseteq G$ a generating subset with $S=S^{-1}$. Let $T$ be a group
acting on $G$ (as group automorphisms) and assume $S$ is invariant under $T$ (i.e. $s^T \subseteq S$, for every $s \in S$).
The semi-direct product group  $G \rtimes T$ acts on the Cayley graph $Cay(G,S)$ as graph automorphism as follows.
For $(g,t) \in G \rtimes T$, $(g',s) \in G \times S$.

$$(g',s)^{(g,t)} = (g g'^{t}, s^{t})$$

We have the following properties that can be easily verified by direct calculations.

\begin{itemize}
\item This is a well defined action, i.e. for $(g_1,t_1), (g_2,t_2) \in G \rtimes T$, $(g',s) \in (G \times S)$,
$$ ((g',s)^{(g_2,t_2)})^{(g_1,t_1)}  = (g',s)^{(g_1,t_1) \cdot (g_2,t_2)} .$$

Note that the undirected edge $(g',s)$ can be also presented as $(g's, s^{-1})$, and indeed:
$$(g's, s^{-1})^{(g,t)} = (g (g's)^t, (s^{-1})^t) = (g g'^{t} s^{t}, (s^t)^{-1}).$$

which represents the same edge as $(g g'^{t}, s^{t})=(g',s)^{(g,t)}$.

\item This action is always transitive on the vertices, and if the action of $T$ on $S$ is transitive then the action of $G \rtimes T$ on $Cay(G,S)$ is edge transitive.

\item This action when restricted to $G$ (sitting as a subgroup $\{(g,1) | g \in G \}$ in $G \rtimes T$) is the usual action of $G$ on $Cay(G,S)$ by multiplication from the left.
\end{itemize}

\end{definition}

\section{Cayley Codes}\label{sec:cayley}

\begin{definition}[Linear code, length, dimension, rate, distance]\label{def:codes}
Let $\F$ be a field and $X$ a finite set. A linear code $C \subseteq \F^X $ is a linear subspace.
The orthogonal space to $C$ is the {\em dual-code} $C^{\bot}$. The {\em length of the code} $C$ is $|X|$.
The {\em dimension of the code} $C$ is its dimension as a subspace. The {\em rate} of $C$, denoted by
$r(C)$, is the dimension of the code divided by its length. The {\em
weight} of $c \in \F^X$, denoted $w(c)$, is the number of
non-zero coordinates in $c$. The {\em normalized distance} of $C$, denoted $\delta(C)$, is
the minimum weight of a non-zero codeword of $C$, divided by the length of
$C$.
\end{definition}

\begin{definition}[Symmetric and Highly Symmetric codes]\label{def:highly-symmetric-code}
Continuing with the notations of Definition~\ref{def:codes}; We say that a code $C$ is {\em symmetric} (or symmetric with respect to $G$) if there is a group $G$ acting {\em transitively} on $X$ such that $C$ is invariant under the induced action of $G$ on $\F^X$. In such a case $C^{\bot}$ is also invariant under $G$ (since $v \cdot w^g = v^{g^{-1}} \cdot w$ for every $g \in G$ and $v,w \in \F^X)$.
We say that $C$ is {\em single-orbit symmetric} if there exists $v \in C^{\bot}$ such that the $G$-orbit $v^{G} \subseteq \F^X$ of $v$ spans $C^{\bot}$. A code $C$ (or more precisely a family of codes $C$ when $|X| \rightarrow \infty$) is said to be {\em highly-symmetric} if $C^{\bot}$ is spanned by an orbit of $v^G$ where $v$ is of bounded weight.
\end{definition}

So, symmetric $C$ means that "all variables" look the same, and it is also single-orbit symmetric if one equation defines $C$ as a symmetric code.
Many of the codes studied classically in coding theory are single-orbit symmetric. For example, all cyclic codes are such (see below and in Section~\ref{section:cyclic}). Let us observe that the action of $G$ on $X$, and hence on $\F^X$, makes $\F^X$ into an $\F[G]$-module, where $\F[G]$ denotes the group algebra of $G$ over $\F$. If $C$ is invariant under $G$, it simply means that $C$ ia an $\F[G]$-submodule. So, given the transitive action of $G$ on $X$, symmetric codes in $\F^X$ are the same as $\F[G]$-submodules. Now, if $C$ is a submodule, so is $C^{\bot}$. The symmetric code/submodule $C$ is single-orbit symmetric iff $C^{\bot}$ is $1$-generated submodule (also called cyclic submodule). Note that an $\F[G]$-module $M$ is $1$-generated iff it is isomorphic to a quotient module of $\F[G]$ (as an $\F[G]$ module).

As $G$ acts transitively on $X$, $X$ can be identified with the coset space $G/H$ of some subgroup $H$ of $G$ and $\F^X = \F^{G/H}$ is a quotient module of $\F^G = \F[G]$, i.e., $\F^X$ is a cyclic module (namely $1$-generated). But here a word of warning is needed. In general an $\F[G]$-submodule of a cyclic $\F[G]$-module is not necessarily cyclic, which means that symmetric codes are not necessarily single-orbit symmetric.

Here is a well known example. Let $\F_p$ be the field of prime order $p$ and $G$ a finite $p$-group. Let $\Delta$ be the augmentation ideal of $\F_p[G]$, i.e., $\Delta=Ker(j)$ where $j:\F_p[G] \rightarrow \F_p$ is defined by $j(\sum a_g g) = \sum a_g$ (here $a_g \in \F_p$ and $g \in G$).
Then $\Delta$ is a submodule (in fact, being the kernel of the homomorphism $j$, it is even a two sided ideal of the group algebra $\F_p[G]$), and it is well known that its minimal number of generators is equal to $d(G)$, the minimal number of generators of $G$ as a group. Now, when $G$ is not a cyclic group, $d(G) \geq 2$ which shows that if we take $X=G$, $\F_p^X$ has $C=\Delta^{\bot}$ as a symmetric code which is {\em not} single-orbit symmetric (since its dual $\Delta= (\Delta^{\bot})^{\bot}$ is not cyclic).  For more information about the number of generators of $\Delta$, its powers and general (left) ideals of $\F_p[g]$ (which are exactly its submodules) see~\cite{Shalev} - especially Section $5$.


On the other hand, if $G$ is a finite cyclic group acting transitively on a set $X$ and $\F$ an arbitrary field, then every submodule of $\F^X$ is cyclic. Indeed $\F^X$ is a quotient of the group algebra $\F[G]$. The last is isomorphic to $\F[t]/(t^n - 1)$ with $n=|G|$, and hence it is a quotient of the polynomial ring $\F[t]$, which is a principle ideal domain and each of its ideals is $1$-generated. So, when $G$ is a cyclic group all the codes which are symmetric w.r.t $G$ are automatically single-orbit symmetric. These are the classical so called "cyclic codes"
(see more in Section~\ref{section:cyclic}).

There is another situation where symmetric codes are automatically single-orbit symmetric; This is when the characteristic of the field $\F$ is prime to the order of $G$. In this case every $\F[G]$ submodule $M$ of $\F^X$ is also a quotient module (since every $\F[G]$-module is semi-simple) and as $F^X$ is $1$-generated so is $M$.

In~\cite{KW}, Kaufman and Wigderson initiated the study of Cayley codes as a method to construct symmetric and highly symmetric codes.
We will use this framework in order to construct constant rate, constant distance LDPC codes whose constraint space is generated by one constraint of constant weight.

\begin{definition}[Cayley code]
Given a Cayley Graph $Cay(G,S)$ and a linear code $B \subseteq \F^S$
of length $|S|=t$  define the linear {\em Cayley code} $Cay(G,S,B)
\subseteq \F^{|G| \cdot |S| / 2}$ as follows. Its coordinates are
the $|G|\cdot|S|/2$  undirected edges of the graph $Cay(G,S)$, namely the pairs $\{(g,s_i),(gs_i,
s_i^{-1})\}$, $g \in G, s_i \in S=\{s_1,\cdots,s_t\}$. The defining linear constraints are the local constraints of $B$ on the edges incident to every vertex, namely

$c \in  Cay(G,S,B)$  iff for every $g \in G$ the following holds:

\begin{itemize}

\item  {\bf Vertex consistency:} $ (c_{\{(g,s_1),(gs_1, s_1^{-1})\}}, \cdots , c_{\{(g,s_t),(g s_t,s_t^{-1})\}}) \in B$.

\end{itemize}

\end{definition}

Assume $T$ is a group acting on $G$ and $S$ is an orbit of $T$, i.e. there exists $\gamma \in G$ s.t. $S=\gamma^T$. Assume further that $S=S^{-1}$.
Let $H=G \rtimes T$. As in Definition~\ref{Def:SDP-action}, $H$ acts on $Cay(G,S)$ and this action is edge transitive. Let now $B \subseteq \F^S$ be a linear code, which is invariant under the action of $T$ (which acts on $S$ and hence on $\F^S$). It is straight forward now to check that $C=C(G,S,B)$ is invariant under $H$ (see also~\cite{KW}).

The following proposition now follows easily from the definition.

\begin{proposition} If $B$ is single-orbit symmetric w.r.t $T$ then $C = C(G,S,B)$ is single-orbit symmetric w.r.t  $H=G \rtimes T$. Moreover, $C^{\bot}$ is generated by the orbit of one vector of weight at most $|S|$.
\end{proposition}

The second statement follows from the fact that each of the constraints defining $C$ is "local" and touches only the variables associated with the edges around a single vertex.

Before bringing the main theorem of this section, let us recall that if $\Gamma$ is an $r$-regular graph of size $m$, then the eigenvalues of its adjacency matrix are $r=\lambda_0 \geq \lambda_1 \geq  \cdots \geq \lambda_{m-1} \geq -r$, and we denote $\lambda=\lambda(\Gamma) = \frac{\lambda_1}{r}$ the second normalized eigenvalue of $\Gamma$.

The following theorem is proved in~\cite{KW}, Theorem $7$. It summarizes Theorem~\ref{thm:Cayley-codes} and Corollary~\ref{cor:Cayley-codes} from the introduction.

\begin{thm}[Detailed Cayley Codes Theorem]\label{thm:Cayley-codes-detailed}
Let $\F$ be a field. Let $G$ and $T$ be groups, such that $T$ acts on $G$, $S \subseteq G$ is an orbit for this action. Assume $S=S^{-1}$ and $S$ generates $G$. The action of $T$ on $S$ induces an action on $\F^S$. Let $B \subseteq \F^S$ be a linear code invariant under $T$.
Assume that:
\begin{itemize}
\item $Cay(G,S)$ is an expander with second normalized eigenvalue
$\lambda$.
\item Normalized distance of $B$ is $\delta > \lambda$.
\item Rate of $B$ is greater than $\frac12$ ($r_B > \frac12$).
\item $B$ is $T$-single-orbit symmetric.
\end{itemize}

Then the code $C(G,S,B)\subseteq \F^{|G||S|/2}$ has constant rate
at least $2r(B) - 1$ and normalized distance at least $[(\delta - \lambda)/(1 - \lambda)]^2$. It
is invariant under the action of the semi-direct product group $H=G
\rtimes T$, and it is $H$-single-orbit symmetric. Moreover, $C(G,S,B)$ is LDPC defined by constraints of weight equal to the one constraint defining $B$ (under the $T$-action).
\end{thm}

\section{Edge Transitive Ramanujan Graphs}\label{section:ramanujan}
In this section we prove Theorem~\ref{theorem:ramanujan}.
This is just a special case of a much more general result in \cite{LSV2} but as the result there is so general, one may find it difficult to see the special case needed here, so we will review it here.  In fact, the special case needed here has already been used in \cite{Lu2} for a different reason and was also explained there. We repeat the description for completeness and also in order to give the explicit description of the set of generators $S$, which amount to give an explicit description of $\gamma$ in the notations of Theorem~\ref{theorem:ramanujan}.

Let ${\F}_q$ and ${\F}_{q^{2}}$ be the fields of order $q$ and $q^2$, say ${\F}_{q^{2}} = {\F}_q [\alpha]$ where $\alpha^2 \in {\F}_q$ is not a square in ${\F}_q$. Following the notations of \cite{LSV2}, we denote by $R$ the ring $R = \F_q[y,\frac{1}{y}, \frac{1}{1+y}]$ , i.e. the subring of the field of rational functions $\F_q(y)$ generated by $y$ , $\frac{1}{1+y}$ and $\frac{1}{y}$. Let $A(R)$ be the four-dimensional $R$-algebra with a basis $1, \alpha, z$ and $\alpha z$ (i.e. it contains the commutative $R$-subalgebra $R[\alpha] = {\F}_{q^{2}}[y,\frac{1}{y},\frac{1}{1+y}]$ as a two-dimensional $R$-module).  The multiplication in $A(R)$ is determined by the rules $z \alpha \ = -\alpha z$ and $z^2 = 1+y$.  As $1+y$ is central and invertible, $z$ is invertible, in fact, $z^{-1} = \frac{1}{1+y}z$. Denote $b = 1 + z^{-1} \in A(R)$.

Now, $A(R)$ contains ${\F}_q[\alpha] = {\F}_{q^{2}}$.  For every $u\in {\F}^*_{q^{2}}$, we denote $\tilde{b_u} = ubu^{-1}$. As ${\F}^*_q$ is in the center of $A(R)$, $\tilde{b_u}$ depends only on the coset of $u$ in ${\F}^*_{q^{2}}/{\F}^*_q$. This gives $\frac{q^2-1}{q-1} = q+1$ elements $\tilde{S}=\{ \tilde{b_u} \mid u\in {\F}^*_{q^{2}}/{\F}^*_q\}$ of $A(R)^*$, where for a ring $D$, we denote by $D^*$ the group of invertible elements.

Let $\tilde\Gamma$ be the subgroup of $A(R)^*$ generated by the $\tilde{b_u}$'s and $\Gamma$ will be its image in $A(R)^*/R^*$, generated by
$S=\{ \tilde{b_u}/{R^*} \mid \tilde{b_u} \in \tilde{S} \}$. For every ideal $I\vartriangleleft R$, we get a map $\pi_I: A(R)^*/R^*\to A(R/I)^*/(R/I)^*$ and we denote the intersection $\Gamma\cap \text{Ker}(\pi_I)$ by $\Gamma (I)$-the congruence subgroup. If $\{0\}\not\cong I\vartriangleleft R$ is a prime ideal and $R/I$ is a finite field of order $q^e$, then $A(R/I)$ (a quaternion algebra) is isomorphic to the $2\times 2$ matrix algebra over ${\F}_{q^{e}}$ (i.e. to $M_2(\F_{q^e})$) and so $A(R/I)^*/(R/I)^*\simeq PGL_2(q^e)$.

Theorem 6.2 of~\cite{LSV2} says that for every $\{ 0\} \not\cong I\vartriangleleft R$, the Cayley graph of $\Gamma /\Gamma (I)$ w.r.t. the generators $S$ (or more precisely the images of $S$ in $\Gamma /\Gamma (I))$ is a $(q+1)$-regular Ramanujan graph. Along the way it is shown there that the set $S$ is symmetric (i.e.  $s\in S$ iff $s^{-1} \in S$).

Now if $I$ is a prime ideal of $R$ with $R/I = {\F}_{q^{e}}$, then $\Gamma /\Gamma (I)$ is isomorphic to a subgroup of $PGL_2(q^e)$, and Theorem 6.6 of \cite{LSV2} shows that it contains $PSL_2(q^e)$. As the latter is a subgroup of index $2$ in the first, $\Gamma/\Gamma(I)$ can be either the first or the second. Moreover, we have enough choices: Theorem 7.1 of \cite{LSV2} ensures that for $q^e > 17$, one can choose $I$ such that $G=PSL_2(q^e)$ or $G=PGL_2(q^e)$ will be obtained. This depends on the image of $b$ : if its image is in $PSL_2(q^e)$ then, as the latter is a normal subgroup of $PGL_2(q^e)$, all its conjugates $\{ b_u\}$ are also in $PSL_2(q^e)$ and vice versa. If the image of $b$ is not in $PSL_2(q^e)$, then the resulting graph is bipartite since $PGL_2(q^e)/PSL_2(q^e)$ is a cyclic group of order $2$, and as said before, all the generators are outside $PSL_2(q^e)$. As it is explained in Corollary 6.8 of \cite{LSV2}, it all depends on the image of $\frac{y}{1+y}$ in $R/I = {\F}_{q^{e}}$. If this image is a quadratic residue there, we will get $PSL_2(q^e)$ and, if it is a non-quadratic residue, we get $PGL_2(q^e)$. The discussion in Section 7 of~\cite{LSV2} shows that at least when $q^e > 17$, there are sufficiently many irreducible polynomials in ${\F}_q[y]$ of degree $e$ to have both possibilities.

Finally, let us observe that since $\F^*_{q^2}$ normalizes $S$, the subgroup ${\F}^*_{q^{2}}$ of $A(R)^*$ normalizes $\tilde{\Gamma}$ and hence also $\Gamma (I)$, so ${\F}^*_{q^{2}}/{\F}^*_q$ also acts on $\Gamma /\Gamma(I)$ and the image of $S$ is the orbit of (the image of) $b$ there. Putting all this information together we see that $Cay(\Gamma/\Gamma(I), S) = Cay(G,S)$ is the promised edge transitive Ramanujan graph.

The above description of the results from  \cite{LSV2} brings only what is needed for this paper. But let us say a few words about the bigger picture (see also Remark 3.6 of \cite{Lu2}, but here we talk only on the case $d=2$, i.e. trees and not general buildings): The group $A(R)^*/R^*$ is a discrete cocompact lattice in $A({\F}_q((y)))^*/{\F}_q((y))^*$. The latter is isomorphic to
$K=PGL_2({\F}_q((y)))$ and it acts on its Bruhat-Tits tree $T$, which is a $(q+1)$-regular tree. The element $b\in K$ takes the initial point of the tree (the vertex $x_0$ corresponding to the lattice ${\F}_q[[y]]^2$) to a vertex $x_1$ of distance one from it. The group ${\F}^*_{q^{2}}/{\F}^*_q$ fixes $x_0$ and acts transitively on the $q+1$ vertices adjacent to $x_0$ (which are in one to one correspondence with the projective line ${\bold P}^1({\F}_q)$ on which indeed ${\F}^*_{q^{2}}/{\F}^*_q$ acts transitively!)  The group $\Gamma$ which is generated by the conjugates $\{ b_u\}$ acts simply transitively on the vertices of $T$. This is a special one-dimensional case of a general result of Cartwright and Steger \cite{CS}. The Cayley graph of $\Gamma$ with respect to $\{ b_u\}$ can therefore be identified with $T$ and hence the Cayley graph $Cay(\Gamma /\Gamma(I), \{ b_u\})$ is isomorphic to $T/\Gamma (I)$.  The fact that the last one is Ramanujan is a deep fact, which follows from the work of Drinfeld (see (\cite{Lu1}, \cite{Mor} and \cite{LSV1}). This is just as in the ``old'' Ramanujan graphs. The extra symmetry that we have in our case is due to the fact that $\Gamma$ is normalized by ${\F}^*_{q^{2}}/{\F}^*_q$ which has order $q+1$ and acts transitively on the $q+1$ generators of $\Gamma$.

\section{Construction of cyclic codes with prescribed parameters}\label{section:cyclic}


%
%
%
%
%

To prove our main theorem, Theorem~\ref{theorem:main}, we will use Theorems~\ref{thm:Cayley-codes-detailed} and~\ref{theorem:ramanujan}, i.e. Cayley codes based on the edge transitive Ramanujan graphs constructed in Theorem~\ref{theorem:ramanujan}. What is left is to construct the "small code" $B$ of Theorem~\ref{thm:Cayley-codes-detailed}. This is what we do now. Specifically, we restate and prove a slightly extended version of Theorem~\ref{theorem:cyclic-code}.

\begin{thm}[Detailed Good $B$ Theorem]\label{theorem:cyclic-code-detailed}
Let $q \in \N$ be a prime power such that $q+1 =2^{m+1} - 2 = 2(2^m-1)$ for $m\geq 10$.
For any constant $a > 2$, there exists an explicit linear binary cyclic code $B \subseteq \F_2^{q+1}$ with $r(B) \geq \frac{1}{2} + \frac{1}{a} > \frac{1}{2}$ and $\delta(B) > \frac{1} {2\log (q+1)(a/(a-2))}$ . In particular, $a$ can be chosen such that $r(B) > \frac{1}{2}$ and $\delta(B) > \frac{2 \sqrt{q}}{q+1}$ (e.g. $a \geq 8$).
\end{thm}

Theorem~\ref{theorem:cyclic-code-detailed} can be deduced from known results in the literature but the requirement that the cyclic code is of even length (i.e. of length $q+1$ for $q$ being a prime power), makes it less routine. Thus, we provide a self contained proof for the existence of such codes.

In the proof we will also review some of the known properties of cyclic codes.
Recall, that a linear code $C \subseteq \F_2^n$ is called {\em cyclic} if for every $(a_0,a_1, \cdots, a_{n-1}) \in C$
also $(a_{n-1},a_0,a_1, \cdots, a_{n-2}) \in C$. To every vector $(a_0,a_1, \cdots, a_{n-1}) \in \F_2^n$ we associate the polynomial $\sum_{i=0}^{n-1}a_ix^i \in \F_2[x]$ and so we can identify $\F_2^n$ with the subspace of polynomials of degree at most $n-1$ in $\F_2[x]$. The latter can be also thought as the elements of the ring $R=\F_2[x]/(x^n-1)$, i.e., $\F_2[x]$ divided by the ideal generated by $x^n-1$. Now, if $C$ is a cyclic linear code in $\F_2^n$, it gives rise to a subspace of $R$ which is invariant under multiplication by $x$, and hence also by its powers. As $C$ is a subspace, it is invariant under multiplication by any element of $R$, i.e. $C$ is an ideal in $R$. Conversely, every ideal of $R$ gives a linear cyclic code in $\F_2^n$.
Now, $\F_2[x]$ is a principle ideal domain and so is $R$. An ideal $C$ in $R$ is thus uniquely defined by its generator $h(x)$ which is a polynomial in $\F_2[x]$ which divides $x^n-1$ (Since in $\F_2[x]$, $(h(x)) \supseteq (x^n-1)$).

Thus, we have a one to one correspondence between linear cyclic codes, ideals in $R$ and divisors of $x^n-1$ in $\F_2[x]$. We can also deduce the following proposition

\begin{proposition}\label{prop:cyclic-code-dim} In the notation above, $\mbox{ dim }C = n-\mbox{deg }h(x)$.
\end{proposition}

We now move to discuss properties of cyclic codes in $\F_2^n$ for $n = 2^m -1$ for some $m \in \N$. These are the most studied cyclic codes.
Let $E = \F_{2^m}$ be the field of order $2^m$. Let $w \in E$ be a primitive element in $E$. Every $\alpha \in E^* = E - \{0\}$ satisfies $\alpha^n =1$. In fact, $E^*$ is exactly the set of all the roots of $x^n-1$.
For $\alpha \in E^*$ denote $m_{\alpha}(x)$ the minimal polynomial of $\alpha$ over $\F_2$, i.e., it is the polynomial of minimal degree in $\F_2[x]$ satisfying $m_{\alpha}(x)=0$. As $\alpha \in E$, which is extension of degree $m$, we know that $\deg(m_{\alpha}(x)) \leq m$.
Also, we have $m_{\alpha}(x) | x^n-1$ since $\alpha$ is a root of $x^n-1$, and for every $\alpha \in E^*$, $m_{\alpha}(x) = m_{\alpha^2}(x)$.

In general, it is not so easy to estimate the distance of the cyclic code $C$ from its generating polynomial $h(x)$. However, this is possible for the following special case that will be used for the proof of Theorem~\ref{theorem:cyclic-code-detailed}.
For $1 \leq r \leq n$, denote $$h_r(x) = l.c.m \{m_{w^i}(x) | 1 \leq i \leq r \}.$$ Since $m_\alpha(x) | x^n -1$, also
$h_r(x) | x^n-1$. As $deg (m_{\alpha}(x)) \leq m$, $deg(h_r(x)) \leq rm$. In fact, for every $\alpha \in E^*$, $m_{\alpha}(x) =m_{\alpha^2}(x)$ and hence $deg(h_r(x)) \leq rm/2$ (for $r \geq 4$). The polynomial $h_r(x)$ gives rise to an ideal (i.e. to a linear cyclic code) $C_r$ called the $BCH(m,r)$ code.  The following proposition provides bounds on the dimension and distance of the $BCH(m,r)$ code.

\begin{proposition}\label{prop:bch} The dimension of the $BCH(m,r)$ code is at least $n-mr/2$  (for $r \geq 4$), and its distance is at least $r+1$.
\end{proposition}

\begin{proof} As the generating polynomial of $BCH(m,r)$ has degree at most $rm/2$ (for $r \geq 4$), we get the bound on the dimension of  $BCH(m,r)$ from Proposition~\ref{prop:cyclic-code-dim}. For obtaining the bound on the distance of $BCH(m,r)$, note that if the code had a codeword $c$ of weight at most $r$, the polynomial associated with this codeword $c(x)$ would be of the following form $c(x) = c_1x^{k_1} +  \cdots + c_r x^{k_r}$ with $k_1 < k_2< \cdots < k_r$. As $h_r(x) | c(x)$ (since the generating polynomial divides every polynomial that is associated with a codeword of the code) we have that $c(w^i) = 0$ for every $i=1, \cdots , r$. This means

$\left(
  \begin{array}{cccccc}
    w^{k_1}  & w^{k_2}   & . & . & . & w^{k_r}  \\
    w^{2k_1} & w^{2k_2}  & . & . & . & w^{2k_r} \\
    . & . & . & . & . & . \\
    . & . & . & . & . & . \\
    . & . & . & . & . & . \\
    w^{rk_1} & w^{rk_2}  & . & . & . & w^{rk_r} \\
  \end{array}
\right)$
$\left(\begin{array}{c}
  c_1 \\
  . \\
  . \\
  . \\
  . \\
  c_r
\end{array}\right)$
$=$
$\left(\begin{array}{c}
  0\\
  .\\
  .\\
  .\\
  .\\
  0
\end{array}\right)$

But the matrix of coefficients has determinant equal to $\prod_{i=1}^r w^{k_i}$ times the determinant of the Vandermonde matrix of $w^1, \cdots, w^r$. Thus, it is not zero. This implies that $c_1 = \cdots = c_r = 0$
\end{proof}



We are now ready to prove Theorem~\ref{theorem:cyclic-code-detailed}.
By Proposition~\ref{prop:bch} the $BCH(m,r)$ code for $r \geq 4$ is a linear cyclic
code of (odd) length $n = 2^m-1$, dimension at least $n-mr/2$  and distance at least $r+1$.
Thus, if we take $r = \lfloor \frac{n}{m}(1 - \frac{2}{a}) \rfloor$ for some constant $a > 2$ we get that $BCH(m,r)$ is a cyclic code of dimension at least $n(1/2 + 1/a)$, i.e with rate at least $1/2+1/a$. Moreover, this code has distance at least $\frac{n}{m}(1 - \frac{2}{a})$, i.e. normalized distance at least $\frac{a-2}{ma}$.

However, for our goal we need the final outcome $B$ to be a cyclic linear code of (even) length $q+1 = 2^{m+1} - 2 = 2(2^m-1) =2n$ over $\F_2$ (where $q$ is a prime power). Using Lemma~\ref{lem:even-cyclic}, we can transform the cyclic code $BCH(m,r)$ of length $n=2^m-1$, described above, to a cyclic code $B$ of length $q+1=2n$, with the same rate as in $BCH(m,r)$, that is, at least $1/2+1/a > 1/2 $, and a normalized distance that is half the normalized distance of $BCH(m,r)$. Namely, a normalized distance which is at least $\frac{a-2}{2ma} > \frac{a-2}{2a \log q}$. Now if $m \geq 10$ and $a \geq 8$
$$\delta(B) \geq \frac{a-2}{2ma} > \frac{a-2}{2a \log q} > \frac{2 \sqrt{q}}{q+1}.$$
So if we choose $m=11$ and $q = 2^{m+1} - 3 = 4093$ all the requirements are satisfied and Theorem~\ref{theorem:main} is proved modulo the following lemma.

\begin{lemma}[A transformation of a cyclic code of odd length to a cyclic code of even length]\label{lem:even-cyclic}
 If there exists a binary cyclic code $C$ of (odd) length $n=2^m-1$, dimension $k$  (i.e. rate $k/n$) and distance $d$ (i.e. a normalized distance of $d/n$), then there exists a binary cyclic code $B$ of (even) length $2n=2(2^m-1)$, dimension $2k$ (i.e. rate $k/n$) and distance $d$ (i.e. a normalized distance of $d/2n$). Moreover, $B$ is explicitly constructed from $C$.
\end{lemma}

\begin{proof} The following proof is essentially from~\cite{VL}. We include it here for completeness.
Let $n=2^m-1$ and $N=2n$. Define a map $\varphi: \F_2^n \times \F_2^n \rightarrow \F_2^N$ by:
$$\varphi (a_0, a_1, \cdots, a_{n-1}, b_0, b_1, \cdots, b_{n-1}) = (a_0, b_1, a_2, b_3, \cdots, b_{n-2}, a_{n-1}, b_0, a_1, b_2, \cdots, a_{n-2}, b_{n-1}).$$
Clearly $\varphi$ is one to one linear map. Denote $B = \varphi(C \times C)$. Thus, $dim(B) = 2k$, and so the rate of $B$, $r(B)$ satisfies $r(B)  = 2k/2n = k/n = r(C)$. It is also easy to see that the distance of $B$ is the same as the distance of $C$ i.e. $d$. This mean that the relative distance of $B$, $\delta(B) = d/2n = \delta(C)/2$. We only need to show that $B$ is a cyclic code, i.e. that $B$ is an ideal when considered as a subspace of the ring $\F_2[x]/(x^N-1)$.
We claim that $B$ is the ideal generated by $h^2(x)$ where $h(x)$ is the generator of $C$ as an ideal of $\F_2[x]/(x^n-1)$.
Note first that $x^N-1 =x^{2n} -1 = (x^n - 1)(x^n + 1) = (x^n-1)^2$ and hence $h^2(x) | x^N-1$, since $h(x) | x^n-1$.
Also, note that $dim(B) = 2 dim(C) = 2(n-deg(h(x))) = N- deg(h^2(x))$.
Thus, it suffices to show that every element of $C$ is divided by $h^2(x)$.
For $a(x), b(x) \in C$ write

\begin{eqnarray*}
a(x) & = & a_0+ a_ 1x + \cdots + a_{n-1}x^{n-1} \\
     & = & (a_0 + a_2 x^2 + \cdots + a_{n-1}x^{n-1}) + x(a_1 + a_3 x^2 + \cdots + a_{n-2}x^{n-3})\\
     & = & a_{even}(x^2) + x a_{odd}(x^2)
\end{eqnarray*}

Note that $n$ is odd. Similarly $b(x) = b_{even}(x^2) + x b_{odd}(x^2)$.
Then $w(x) = \varphi(a(x), b(x))$ can be written as

\begin{eqnarray*}
w(x) & =  & ( a_{even}(x^2)  + x^{n+1}a_{odd}(x^2)) + (x b_{odd}(x^2) + x^n b_{even}(x^2) ) \\
     & =  & (a(x) + x(x^{n} + 1)a_{odd}(x^2)) + (b(x) +  (x^{n} + 1)b_{even}(x^2))
\end{eqnarray*}

Both terms $F(x)= (a(x) + x(x^{n} + 1)a_{odd}(x^2))$ and $G(x)=(b(x) +  (x^{n} + 1)b_{even}(x^2))$ are divisible by $h(x)$,
since $h(x)$ divides $a(x),b(x)$ and $x^n-1 = x^n + 1$.
Also we have that $F(x) = (a(x) + x(x^{n} + 1)a_{odd}(x^2)) = ( a_{even}(x^2)  + x^{n+1}a_{odd}(x^2))$ contains only even powers of $x$,
and $G(x) = (b(x) +  (x^{n} + 1)b_{even}(x^2)) = (x b_{odd}(x^2) + x^n b_{even}(x^2) )$ contains only odd powers of $x$.
This implies that both terms are actually divisible by $h^2(x)$. Indeed, $F(x) = f^2(x)$ for some polynomial $f(x) \in \F_2[x]$ and $G(x) = xg^2(x)$ for some polynomial $g(x) \in \F_2[x]$.
Now, $h(x) | F(x)$ and $h(x) |G(x)$, also, $0$ is not a root of $h(x)$ and all the roots of $h(x)$ are of multiplicity one (here we use the fact that $n$ is odd). Thus, we deduce that $h(x)$ divides $f(x)$ and $g(x)$, which implies that $h^2(x)$ divides $F(x)$ and $G(x)$ and hence $w(x)$.

%
\end{proof}

In summary we have proved:

\begin{thm}[A highly symmetric LDPC good code]\label{theorem:main-with-params}
Let $q=4093$, $\alpha \in \N$, $G=PSL_2(q^{\alpha})$ or  $G=PGL_2(q^{\alpha})$, $T$ the group $T=\F^*_{q^2}/\F^*_{q} \leq PGL_2(q)$, $\gamma \in G$ as in Section~\ref{section:ramanujan}, $S=\{ t \gamma t^{-1} | t \in T \}$. Let $B \in \F^S$ be the code defined in Section~\ref{section:cyclic}, with $a \geq 8$. Then $C(G,S,B)$  is a linear code of rate at least $\frac{2}{a}$,
normalized distance at least $$ \left(\frac{\frac{a-2}{2a \log q} - \frac{2\sqrt{q}}{q+1} }{1 - \frac{2\sqrt{q}}{q+1} }\right)^2.$$
The code is symmetric with respect to the action of the
semi direct product group $G \rtimes T$, and $C^{\bot}$ is generated by  the orbit of a single constraint of weight at most $q+1$.
In particular, for $\alpha \rightarrow \infty$, this is a family of highly symmetric LDPC good codes.

\end{thm}

\newpage



\begin{thebibliography}{99}
\bibitem{BSS05}
L. Babai, A. Shpilka and D. Stefankovic, {\em
Locally testable cyclic codes},IEEE Transactions on Information
Theory, Vol 51, No 8, 2849--2858. 2005.

\bibitem{BERMAN}
S. D. Berman.
{\em Semisimple Cyclic and Abelian Codes.}
Cybernetics 3, 21–-30, 1967.

\bibitem{CS} D. I. Cartwright and T. Steger. {\em  A family of $\tilde A_n$-groups}, Israel J. Math, Vol 103,
125--140, 1998.

\bibitem{KW}
T. Kaufman and A. Wigderson {\em Symmetric LDPC and local Testing}, Innovations in Computer Science, 406--421, 2010.

\bibitem{Lu1} A. Lubotzky {\em Discrete groups, expanding graphs and invariant measures}.
With an appendix by Jonathan D. Rogawski. Reprint of the 1994 edition, Modern Birkhauser Classics, Birkhauser Verlag, Basel, iii + 192, 2010.

\bibitem{Lu2} A. Lubotzky {\em Simple groups of Lie type as expanders},
J. of the European Math. Soc., to appear.

\bibitem{LPS}
A. Lubotzky, R. Phillips and P. Sarnak, {\em Ramanujan graphs}, Combinatorica 8, no. 3, 261--277, 1988.

\bibitem{LSV1}
A. Lubotzky, B. Samuels and U. Vishne, {\em Ramanujan complexes of type $\tilde A_d$}
Israel J. Math. Vol 149, 267--299, 2005.


\bibitem{LSV2}
A. Lubotzky, B. Samuels and U. Vishne, {\em Explicit constructions of Ramanujan
complexes of type $\tilde{A}_d$}, European J. Combin. 26 , no. 6, 965--993, 2005.


\bibitem{Mor}
M. Morgenstern, {\em Existence and explicit constructions of $q+1$ regular Ramanujan graphs for every prme power q},
Journal of Combinatorial Theory, Series B 62, 44--62, 1994.

\bibitem{Shalev}
A. Shalev, {\em Subgroups, nilpotency indices, and the
number of generators of ideals in $p$-group algebras}, J. Algebra 129, no. 2, 412--438, 1990.

\bibitem{SS96}
M. Sipser and D. A. Spielman, {\em Expander codes}, IEEE
Transactions on Information Theory, Vol 42, No 6, 1710--1722, 1996.


\bibitem{Tan81} R. M. Tanner, {\em A recursive approach to low complexity codes}, IEEE
Transactions on Information Theory, 27(5):533--547, 1981.


\bibitem{VL}
J. H. van Lint {\em Repeated-Root Cyclic Codes}
IEEE Transactions on Information Theory, Vol 37, No 2, 343--345, 1991.


\end{thebibliography}

\begin{thebibliography}{99}

\bibitem{AKKLR}  Noga Alon, Tali Kaufman, Michael Krivelevich, Simon Litsyn and Dana
Ron, {\em Testing Low Degree Polynomials Over GF(2)}, Proceedings of
7th International Workshop on Randomization and
Computation,(RANDOM), Lecture Notes in Computer Science 2764,
188-199, 2003. Also, IEEE Transactions on Information Theory, Vol.
51(11), 4032-4039, 2005.


\bibitem{ALW}
Noga Alon, Alex Lubotzky and Avi Wigderson {\em Semi Direct product
in groups and zig-zag product in graphs: connections and
applictions}, Proceedings of the 42nd Annual Symposium on the
Foundations of Computer Science (FOCS), 630-637, 2001.

\bibitem{AS}
Sanjeev Arora and Madhu Sudan.
{\em Improved low degree testing and its applications.}
Combinatorica, 23(3): 365-426, 2003.




\bibitem{BFL}
L. Babai and L. Fortnow and C. Lund, {\em Non-Deterministic
Exponential Time has Two-Prover Interactive Protocols},
Computational Complexity, volume 1, number 1, 3--40, 1991.

\bibitem{BSS05}
L{\'a}szl{\'o} Babai and Amir Shpilka and Daniel Stefankovic, {\em
Locally testable cyclic codes},IEEE Transactions on Information
Theory, Vol 51, No 8, pp. 2849--2858. 2005.



\bibitem{BHR05}
Eli Ben-Sasson, Prahladh Harsha and Sofya Raskhodnikova, {\em Some
3CNF Properties are Hard to Test}, SIAM Journal on Computing, volume
35, issue 1, pages 1-21, 2005.


\bibitem{BS05}
Eli Ben-Sasson, Madhu Sudan, {\em Simple PCPs with poly-log rate and
query complexity}, STOC 2005: 266-275.



\bibitem{BSVW03}
E. Ben-Sasson, M. Sudan, S. Vadhan, A. Wigderson.
{\em Randomness-efficient Low Degree Tests and Short PCPs via Epsilon-Biased Sets}
35th Annual ACM Symposium, STOC 2003, pp. 612-621, 2003.

\bibitem{BLR}
Blum, M., Luby, M., Rubinfeld, R., {\em Self-Testing/Correcting with
Applications to Numerical Problems}, In J. Comp. Sys. Sci. Vol. 47,
No. 3, December 1993.



\bibitem{CRVW}
M. Capalbo, O. Reingold, S. Vadhan, A. Wigderson, {\em Randomness
Conductors and Constant-Degree Expansion Beyond the Degree /2
Barrier}, Proceedings of the 34th STOC, pp. 659-668, 2002.



\bibitem{CU}
L. Carlitz and S. Uchiyama, {\em Bounds for exponential sums}, Duke
Math. J., 24:37-41, 1957.


\bibitem{Dinur}
Irit Dinur,{\em The PCP theorem by gap amplification}, J. ACM 54(3):
12 (2007).


\bibitem{GAL}
R. G. Gallager, {\em Low density parity check codes}, MIT Press,
Cambridge, MA, 1963.

\bibitem{GKS09}
Elena Grigorescu, Tali Kaufman and Madhu Sudan, {\em Succinct
Representation of Codes with Applications to Testing}, manuscript.


\bibitem{GS}
Oded Goldreich, Madhu Sudan, {\em Locally testable codes and PCPs of
almost-linear length}, J. ACM 53(4): 558-655 (2006).

\bibitem{HS}
Holton, D. A. and Sheehan, J. {\em The Petersen Graph}. Cambridge,
England, Cambridge University Press, 1993.

\bibitem{KS:invariant}
Tali Kaufman and Madhu Sudan, {\em Algebraic Property Testing: The
Role of Invariance}, Proceedings of the 40th ACM Symposium on Theory
of Computing (STOC), 2008.

\bibitem{KL05}
Tali Kaufman, Simon Litsyn, {\em Almost Orthogonal Linear Codes are
Locally Testable}, FOCS 2005: 317-326.

\bibitem{Lackenby1} M. Lackenby, {\em Large groups, property $(\tau)$ and the homology growth of subgroups.}
Math. Proc. Cambridge Philos. Soc. 146 (2009), no. 3, 625--648.

\bibitem{Lackenby2} M. Lackenby, {\em Covering spaces of 3-orbifolds.} Duke Math. J. 136 (2007), no. 1, 181--203.



\bibitem{LMSS01}
Michael G. Luby, Michael Mitzenmacher, M. Amin Shokrollahi and
Daniel A. Spielman, {\em Improved Low-Density Parity-Check Codes
Using Irregular Graphs} IEEE Transactions on Information Theory,
47(2), pp. 585-598. 2001.


\bibitem{LW93}
A. Lubotzky, B. Weiss, {\em Groups and expanders}, In {\em Expanding
Graphs} (e. J. Friedman), DIMACS Ser. Discrete Math. Theoret. Compt.
Sci. 10pp. 95-109, Amer. Math. Soc., Prividence, RI 1993.



\bibitem{MS}
F. J. MacWilliams and N. J. A. Sloan, {\em The Theory of Error
Correcting Codes}, North Holland, Amsterdam, 1977.


\bibitem{OrMeir}
Or Meir, {\em Combinatorial Construction of Locally Testable Codes},
proceedings of STOC 2008, pages 285-294.



\bibitem{MW}
R. Meshulam, A. Wigderson, {\em Expanders in Group Algebras},
Combinatorica, vol. 24, no. 4, pp 659-680, 2004.



\bibitem{RS}
Ronitt Rubinfeld and Madhu Sudan, {\em Robust characterizations of
polynomials with applications to program testing}, SIAM Journal on
Computing, 25(2):252-271, April 1996.

\bibitem{RSW}
E. Rozenman, A. Shalev, A. Wigderson, {\em A new family of Cayley
expanders (?)}, 36th Annual ACM Symposium, STOC 2004, pp. 445-454,
2004.



\bibitem{RU01}
T. Richardson and R. Urbanke, {\em The Capacity of Low-Density
Parity Check Codes under Message-Passing Decoding}, IEEE
Transactions on Information Theory, 47(2):599-618, 2001.


\bibitem{RVW}
O. Reingold, S. Vadhan, A. Wigderson, {\em Entropy Waves, the
Zig-Zag Graph Product, and New Constant-Degree Expanders}, Annals of
Mathematics, vol. 155, no.1, pp. 157-187, 2002.


\bibitem{SS96}
Michael Sipser  and Daniel A. Spielman, {\em Expander codes}, IEEE
Transactions on Information Theory, Vol 42, No 6, pp. 1710-1722.
1996.





\bibitem{MadhuNotes}
\newblock Madhu Sudan Lecture notes
\newblock http://people.csail.mit.edu/madhu/FT01/scribe/bch.ps.

\bibitem{STV}
Madhu Sudan, Luca Trevisan, and Salil Vadhan, {\em Pseudorandom generators without the XOR Lemma},
Journal of Computer and System Sciences, 62(2): 236--266, March 2001.


\bibitem{Tan81} Robert M. Tanner, {\em A recursive approach to low complexity codes}, IEEE
Transactions on Information Theory, 27(5):533-547, 1981.



\bibitem{WEIL}
A. Weil, {\em Sur les courbes algebriques et les varietes qui s'en
deduisent}, Actualities Sci. et Ind. no. 1041. Hermann, Paris, 1948.

\end{thebibliography}
\end{document}